\documentclass[11pt,letterpaper]{article}

\usepackage[margin=1in]{geometry}
\usepackage[T1]{fontenc}
\usepackage{lmodern}
\usepackage{microtype}

\usepackage{amsmath,amsfonts,amssymb}
\usepackage{amsthm}
\usepackage{mathtools}
\usepackage{physics}
\usepackage{orcidlink}
\usepackage{array}
\usepackage{multirow}
\usepackage{graphicx}
\usepackage{float}
\usepackage{textcomp}
\usepackage{url}
\usepackage{cite}
\usepackage{authblk}
\usepackage[font=normalsize]{caption}
\newenvironment{IEEEkeywords}%
{\par\medskip\noindent\textbf{Keywords: }}%
{\par}
\newcommand{\appendices}{\appendix}

\newcolumntype{C}[1]{>{\centering\arraybackslash}m{#1}}

\newcommand{\Id}{{\rm 1\hspace{-0.9mm}l}}

\newcommand{\PP}{\ensuremath{\mathcal{P}}}
\newcommand{\QQ}{\ensuremath{\mathcal{Q}}}

\renewcommand{\AA}{\ensuremath{\mathcal{A}}}

\newcommand{\proj}[1]{\ensuremath{\ketbra{#1}{#1}}}
\newcommand{\ketV}[1]{\ensuremath{|#1\rangle\!\rangle}}
\newcommand{\braV}[1]{\ensuremath{\langle\!\langle#1|}}
\newcommand{\ketbraV}[2]{\ensuremath{\ketV{#1}\!\braV{#2}}}
\newcommand{\projV}[1]{\ensuremath{\ketbraV{#1}{#1}}}

\newtheorem{theorem}{Theorem}
\newtheorem{corollary}[theorem]{Corollary}
\newtheorem{proposition}[theorem]{Proposition}

\begin{document}

\title{Noise Resilience of Quantum Key Distribution Protocols Secured Against Independent Attacks With One-Way Communication}

\author[1,2]{Adam B\'ilek\,\orcidlink{0000-0001-8380-5338}\thanks{Corresponding author: Adam B\'ilek (e-mail: adam.bilek@vsb.cz).}}
\author[2]{Ryszard Kukulski\,\orcidlink{0000-0002-9171-1734}}
\author[2]{Paulina Lewandowska\,\orcidlink{0000-0003-1564-7782}}
\author[3]{{\L}ukasz Pawela\,\orcidlink{0000-0002-0476-7132}}
\author[3]{Zbigniew Pucha{\l}a\,\orcidlink{0000-0002-4739-0400}}
\affil[1]{Department of Applied Mathematics, Faculty of Electrical Engineering and Computer Science, VSB--Technical University of Ostrava, 17.~listopadu 2172/15, 708 33 Ostrava, Czech Republic}
\affil[2]{IT4Innovations, VSB--Technical University of Ostrava, 17.~listopadu 2172/15, 708 33 Ostrava, Czech Republic}
\affil[3]{Institute of Theoretical and Applied Informatics, Polish Academy of Sciences, Ba{\l}tycka 5, 44-100 Gliwice, Poland}

\date{10 July 2026}

\markboth{IEEE Transactions on Information Theory,~Vol.~XX, No.~XX, Month~Year}%
{B\'ilek \MakeLowercase{\textit{et al.}}: Noise Resilience of QKD Protocols Secured Against Independent Attacks With One-Way Communication}

\maketitle

\begin{abstract}
    We investigate the resilience to noise of single-qubit quantum key distribution (QKD) protocols in the scenario of security against independent eavesdropping attacks and key distillation based on one-way classical communication. To this end, we introduce a noise-based metric that quantifies the efficiency of QKD protocols. Within this framework, we analyze the maximal noise levels that allow Alice and Bob to asymptotically establish a secure secret key. Using this assumption, we compare the noise tolerance of general single-qubit QKD protocols, in particular the BB84, B92, E91, and six-state protocols.
    Our main result determines the noise level threshold for QKD allowing one to distill an asymptotically secure secret key. Additionally, we demonstrate that the six-state protocol achieves the greatest resistance to noise while simultaneously yielding a higher post-selection efficiency than the other analyzed single-qubit protocols, confirming its robustness within the considered security model.  Finally, we perform an analysis of the proposed noise-based metric and the conventional quantum bit error rate (QBER) metric.
\end{abstract}


\section{Introduction}\label{sec:Introduction}
Quantum key distribution (QKD) enables two distant parties to establish a shared secret key with security guaranteed by the laws of quantum mechanics. Since the introduction of the BB84~\cite{bb84,bennett2014quantum}, E91~\cite{ekert1991quantum}, and B92~\cite{Bennett_1992} protocols, QKD has become one of the most mature applications of quantum information science. Its security relies on fundamental quantum principles, including the impossibility of perfectly cloning unknown quantum states and the disturbance inevitably introduced by quantum measurements.

While ideal QKD protocols provide information-theoretic security, practical implementations are inevitably affected by imperfections in state preparation, quantum transmission, and measurement. These imperfections reduce the achievable secret key rate and may compromise security if not properly taken into account~\cite{slutsky1998security,lutkenhaus2000security}. Consequently, practical QKD systems are typically evaluated using performance indicators such as the quantum bit error rate (QBER) and the secret key rate (SKR), which quantify empirical errors and the rate of distilled secret key, respectively.

The influence of eavesdropping strategies and implementation imperfections on QKD security has also been extensively investigated.
Early works analyzed optimal individual attacks against the BB84~\cite{fuchs1997optimal} and six-state~\cite{bruss1998optimal, bechmann1999incoherent} protocols, while subsequent studies considered increasingly realistic implementation assumptions and practical attack scenarios~\cite{waks2006security,shapiro2006attacking}. Comprehensive treatments of practical QKD security are given in~\cite{scarani2009security}, and recent research has further extended these analyses to realistic device models and modern security frameworks~\cite{ramanathan2023security,roy2024security}. Information-theoretic security proofs with one-way classical post-processing, together with the corresponding QBER thresholds for the BB84 and six-state protocols, were established in~\cite{renner2005information}.

Despite these advances, a fundamental question remains barely explored: how should the overall physical noise affecting a QKD implementation be quantified? In this work, we challenge the QBER measure and show that it does not necessarily reflect the actual level of physical noise experienced by the quantum communication process. Consequently, empirical error rates alone cannot always be interpreted as a direct measure of protocol noise resilience.

To address this problem, we introduce a noise metric, denoted by $\epsilon$, which quantifies the cumulative deviation of a practical implementation from its ideal realization. The proposed metric captures imperfections originating from all stages of a QKD protocol, including state preparation, quantum transmission, measurement, and information leakage caused by an eavesdropper. Building upon this metric, we define the corresponding noise threshold $\Upsilon$, representing the maximum admissible physical noise under which secure key generation remains possible. Beyond introducing the $\epsilon$-based framework, we demonstrate that QBER alone is not always a reliable indicator of protocol robustness. In particular, we construct a family of prepare-and-measure protocols employing highly non-orthogonal quantum states that tolerate remarkably large empirical error rates in the generated bitstrings while remaining subject to only a small amount of physical noise, as quantified by the proposed metric $\Upsilon$. This example illustrates that empirical bit errors and physical implementation noise are fundamentally distinct quantities, highlighting the need for protocol-independent measures of noise resilience. Throughout most of the work, we consider a general class of prepare-and-measure QKD protocols operating in independent rounds, with arbitrary independent attacks allowed and one-way communication. Within this framework, we formulate an optimization problem over all admissible protocols and characterize the highest noise threshold compatible with a positive asymptotic secret key rate.

The remainder of this paper is organized as follows. In Section~\ref{sec:problem}, we formulate the problem and introduce the theoretical framework underlying our analysis. Section~\ref{sec:results} presents the main analytical and numerical results on the noise resilience of QKD protocols. In particular, Subsection~\ref{sub-noise} derives the fundamental threshold of the noise resilience achievable by general QKD protocols, while Subsection~\ref{sub-qber}  presents the motivation of our work by examining  the discrepancy between QBER and  the $\Upsilon$ measure. Section~\ref{sec:discuss} concludes the paper and outlines directions for future research.

\section{Problem Definition}\label{sec:problem}

Let us assume the following scenario. Alice and Bob define a quantum key distribution protocol $\mathcal{Q}$ consisting of a sequence of independent rounds. In each round, Alice prepares and sends a qubit state to Bob. Subsequently, the parties may perform public classical communication and private processing of their local quantum registers. At the end of the round, they either accept the outcome (event $S$) or discard it. By
$\mathcal{P}(S)$ we will
denote the corresponding probability of accepting the round. Whenever the $i^\mathrm{th}$ round is accepted, Alice and Bob obtain classical bits $a_i$ and $b_i$, respectively. After many repetitions, they accumulate two bitstrings, $(a_i)_i$ and $(b_i)_i$. The protocol uses random bit flipping to enforce uniform marginals and employs one-way communication to distill a secure secret key via error correction and privacy amplification.

We assume that the quantum communication is intercepted by an eavesdropper, Eve, who is restricted to arbitrary independent attacks. Within each communication round, Eve's attack is fully general: she may store a quantum system in her quantum memory and postpone its measurement until after the public communication associated with that round. Based on the intercepted quantum communication and the information overheard from public communication,
Eve generates her bitstring $(e_i)_i$. In that scenario, each round can be treated independently, and then the asymptotic secret key rate (SKR) can be analyzed utilizing the wiretap channel given by
\cite{DevetakWinter2005}
\begin{equation}
    \mathrm{SKR} =
    \max_{P\in \{A,B\}}
    \left[
        I(A;B)-I(P;E)
        \right].
\end{equation}

The maximization reflects the choice of reconciliation direction. Due to the symmetrization procedure~\cite{Shannon1948,CoverThomas2006}, the above expression simplifies to

\begin{equation}
    \mathrm{SKR} =
    \max_{P\in\{A,B\}}
    \left[
        H_2 \left(\PP(P=E|S)\right) -
        H_2 \left(\PP(A=B|S)\right)
        \right],
\end{equation}
where $H_2$ denotes the binary entropy function \cite{Shannon1948} and $\PP(P_1=P_2|S)$ is the probability that the accepted bits of parties $P_1$ and $P_2$ coincide in a single round.

The objective of this work is to quantify the noise resilience of QKD protocols within the provided framework. The intensity of noise will be parametrized by $\epsilon\in[0,1]$. As discussed in the introduction, $\epsilon$ captures all deviations from the ideal implementation, including state-preparation of Alice, Bob's measurement, noisy quantum communication line between Alice and Bob and, finally, the effects of independent eavesdropping attacks. Collectively, these effects are described by a quantum channel $\Phi_\epsilon$ satisfying
\begin{equation}\label{fidelity}
    \frac14
    \braV{\Id}
    J(\Phi_\epsilon)
    \ketV{\Id}
    \ge
    1-\epsilon,
\end{equation}
where $J(\Phi_\epsilon)$ denotes the Choi--Jamio{\l}kowski representation of the channel $\Phi_\epsilon$ \cite{watrous2018theory} and $\ketV{\Id}$ denotes the vectorization of the identity operator. Thus, $\Phi_\epsilon$ is constrained to be at most $\epsilon$ away from the identity channel as measured by the average entanglement fidelity \cite{Nielsen2002, horodecki1999general}.
The disturbances induced by $\Phi_\epsilon$ are observable to Alice and Bob through the quantum bit error rate (QBER) \cite{gisin2002quantum, bennett2014quantum}, or equivalently through the quantity $\PP(A=B|S)$. Whenever $\epsilon=0$, there is no disturbance to $\QQ$ as $\Phi_0 = \Id$.  In this noiseless scenario, we assume two properties: Alice and Bob always agree correctly on the value of the key, that is $\PP(A=B|S)=1$; Eve has no knowledge about the secret value of key, that is $\PP(A=E|S)=1/2$.

We say that a protocol $\mathcal{Q}$ tolerates a noise level $\epsilon$ if it enables asymptotic secret key generation secure against independent eavesdropping attacks using one-way classical communication. To formalize this notion, we introduce two functions, $x_{\mathcal{Q}}(\epsilon)$ and $y_{\mathcal{Q}}(\epsilon)$. The function $x_{\mathcal{Q}}(\epsilon)$ characterizes the worst-case acceptable rate for the QKD protocol $\mathcal{Q}$ with intensity of noise $\epsilon$ as
\begin{equation}
    x_{\mathcal{Q}}(\epsilon) = \max(\min_{\Phi_\epsilon} \PP(A=B|S), 1/2).
\end{equation}
In particular, it holds $x_\QQ(0) = 1$ for any $\QQ$. Observe that in the anti-correlation regime $\PP(A=B|S) < 1/2$, Alice and Bob could raise their compatibility by using public communication up to $1 - \PP(A=B|S) > 1/2$. Hence, the maximum in the definition guarantees that if the noise channel $\Phi_\epsilon$ can achieve $\PP(A=B|S) < 1/2$, from convexity it would be able to find the noise channel such that $x_{\mathcal{Q}}(\epsilon) = 1/2$, which is the most destructive interference. The function $y_{\mathcal{Q}}(\epsilon)$ quantifies the maximal information that an eavesdropper employing an independent attack can obtain while preserving at least the agreement level $x_{\mathcal{Q}}(\epsilon)$
\begin{equation}
    y_{\mathcal{Q}}(\epsilon) =
    \min_{P\in \{A,B\} }
    \max_{E}
    \left\{
    \PP(P=E|S):
    \PP(A=B|S)\ge x_{\mathcal{Q}}(\epsilon)
    \right\},
\end{equation}
where the maximization is taken over any single-round attack strategy of Eve, that is, over all combinations of an intercepting quantum channel and conditional measurements used to produce her bit $E$. A formal parametrization of such strategies is provided in Appendix~\ref{proof-1-1}. In particular, it holds $y_\QQ(0) = 1/2$ for any $\QQ$. Within our framework, the protocol $\mathcal{Q}$ achieves a positive asymptotic secret key rate SKR$>0$ if and only if
\begin{equation}
    y_{\mathcal{Q}}(\epsilon)
    <
    x_{\mathcal{Q}}(\epsilon).
\end{equation}
This motivates the definition of the maximal tolerable noise level of protocol $\mathcal{Q}$,
\begin{equation}
    \Upsilon_{\mathcal{Q}}
    =
    \sup
    \left\{
    \epsilon\in[0,1]:
    y_{\mathcal{Q}}(\epsilon)
    <
    x_{\mathcal{Q}}(\epsilon)
    \right\}.
\end{equation}
Finally, our ultimate goal is to determine the highest noise level threshold achievable by any QKD protocol $\mathcal{Q}$ against independent eavesdropping attacks. We therefore define the noise threshold achievable by any $\mathcal{Q}$ as
\begin{equation}
    \Upsilon = \sup_{\mathcal{Q}}
    \Upsilon_{\mathcal{Q}}.
\end{equation}

\section{Results}\label{sec:results}
This section presents the main results of this work. The proofs of theorems are presented in the Appendix.

\subsection{Noise resilience}\label{sub-noise}

\begin{theorem}\label{thm-1-1}
    Under the assumptions of Section~\ref{sec:problem}, the noise level threshold for QKD protected against independent attacks with the usage of one-way communication is
    \begin{equation}
        \Upsilon=\frac14.
    \end{equation}
\end{theorem}
The proof of Theorem~\ref{thm-1-1} is presented in Appendix~\ref{proof-1-1}.

For now, let us consider a particular class of QKD schemes, namely prepare-and-measure protocols~\cite{shu2023asymptotically}. Let Alice be given a collection of probabilities
$
    (p_{i,a})_{\substack{i=1,\ldots,N\\ a=0,1}}, $
where \(p_{i,a}\geq 0\), together with a corresponding set of qubit states
$
    \{\rho_{i,a}\}_{\substack{i=1,\ldots,N\\a=0,1}}. $
We assume that, for every \(i\), $
    p_{i,0} + p_{i, 1} > 0 $.
In each round, Alice prepares the state \(\rho_{i,a}\) with probability \(p_{i,a}\) and sends it to Bob.
Bob is equipped with a sub-normalized measurement, $
    \{\Omega_{i,b}\}_{\substack{i=1,\ldots,N\\b=0,1}},$ that means, $ \sum_{i,b} \Omega_{i,b} \le \Id$,
where \(\Omega_{i,b}\geq 0\) for all \(i\) and \(b\), such that
$
    \Omega_{i,0} + \Omega_{i, 1} > 0
$
for every $i$.
The indices $a$ and $b$ denote the classical bit values associated with Alice's preparation choice and Bob's measurement outcome, respectively.

A prominent example of a prepare-and-measure QKD protocol is the six-state protocol \(\QQ_{6S}\) introduced in~\cite{bruss1998optimal,bechmann1999incoherent}, corresponding to the case \(N=3\). In this protocol, Alice chooses each state with equal probability, that is,
\[
    \forall_{i, a} \;\; p_{i,a}=\frac{1}{6}
\]
for all \(i\in\{1,2,3\}\) and \(a\in\{0,1\}\). The signal states are pure and are given by
\[
    \ket{\psi_{1,0}}=\ket{0}, \qquad
    \ket{\psi_{1,1}}=\ket{1},
\]
\[
    \ket{\psi_{2,0}}=\ket{+}, \qquad
    \ket{\psi_{2,1}}=\ket{-},
\]
\[
    \ket{\psi_{3,0}}=\ket{i}, \qquad
    \ket{\psi_{3,1}}=\ket{-i}.
\]
These states constitute the eigenstates of the three Pauli observables. Bob performs the corresponding measurements in the same bases, described by the POVM elements $
    \Omega_{i,b}=\frac{1}{3}\proj{\psi_{i,b}}, $
for \(i\in\{1,2,3\}\) and \(b\in\{0,1\}\).
It is worth observing that the prepare-and-measure version of the six-state protocol achieves the post-selection efficiency $\PP(S)=\frac13.$
The six-state protocol plays a distinguished role in our analysis due to the following Corollary.

\begin{corollary}\label{cor-1-2}
    For the six-state QKD protocol, we have $
        \Upsilon_{\QQ_{6S}}=\Upsilon=\frac14. $
\end{corollary}
The proof of Corollary~\ref{cor-1-2} is presented in Appendix~\ref{proof-1-2}.

Several notable QKD protocols have been proposed, including BB84, B92, and E91. These protocols can be compared using our $\epsilon$-metric by computing $\Upsilon_\QQ$ via semidefinite programming (SDP)~\cite{vandenberghe1996semidefinite}.
To optimize this problem, we used the
\texttt{Julia}
programming language~\cite{julia} and SDP optimization via SCS
solver~\cite{ocpb:16, scs} with absolute convergence tolerance $10^{-5}$. The code is available on GitHub~\cite{repository}.
To determine whether $y_\QQ(\epsilon) < x_\QQ(\epsilon)$, we formulate two SDP programs based on the proof of Theorem~\ref{thm-1-1} presented in Appendix~\ref{proof-1-1}: SDP~\ref{sdp-1} for computing $x_\QQ(\epsilon)$ and SDP~\ref{sdp-2} for computing $y_\QQ(\epsilon)=\min_{P=A,B} y_{\QQ,P}(\epsilon)$. Owing to the monotonicity of the problem, the threshold value $\Upsilon_\QQ$ can then be efficiently determined using a standard bisection method.
\begin{table}[!t]
    \centerline{\underline{SDP program for $x_\QQ(\epsilon)$ -- primal problem}}
    \begin{equation*}
        \begin{split}
            \text{minimize:}\quad &
            \tr\left( J \sum_{\substack{i=1,\ldots,N                             \\ a = 0,1}} p_{i,a} \overline{\rho_{i,a}} \otimes \Omega_{i,a}\right)
            \\[2mm]
            \text{subject to:}\quad
                                  & \tr\left( J \sum_{\substack{i=1,\ldots,N     \\ a,b = 0,1}} p_{i,a} \overline{\rho_{i,a}} \otimes \Omega_{i,b}\right)=1, 	\\
                                  & \braV{\Id}J\ketV{\Id}\ge\tr(J)(2-2\epsilon), \\
                                  & \tr_2(J) \propto \Id,                        \\
                                  & J \ge 0,                                     \\[2mm]
            \text{variable:}\quad & J \in \mathcal{M}_{4\times4}(\mathbb{C})
        \end{split}
    \end{equation*}
    \captionsetup{name=SDP}
    \caption{Semidefinite program for computing $x_{\mathcal{Q}}$. }
    \label{sdp-1}
\end{table}

\begin{table}[!t]
    \centerline{\underline{SDP program for $y_{\QQ,P}(\epsilon)$ -- primal problem}}
    \begin{equation*}
        \begin{split}
            \text{maximize:}\quad  &
            \sum_{\substack{i=1,\ldots,N                                                                                                                 \\ a,b = 0,1}} p_{i,a} \tr\left( E_{i,p(a,b)} \overline{\rho_{i,a}} \otimes \Omega_{i,b}\right)
            \\[2mm]
            \text{subject to:}\quad
                                   & \sum_{\substack{i=1,\ldots,N                                                                                        \\ a,p = 0,1}} p_{i,a} \tr\left( E_{i,p} \overline{\rho_{i,a}} \otimes \Omega_{i,a}\right) \ge x_\QQ(\epsilon), \\
                                   & \sum_{\substack{i=1,\ldots,N                                                                                        \\ a,b,p = 0,1}} p_{i,a} \tr\left( E_{i,p} \overline{\rho_{i,a}} \otimes \Omega_{i,b}\right) = 1, \\
                                   & \sum_{p=0,1} E_{i,p} = \sum_{p=0,1} E_{j,p} \quad \mathrm{for} \quad i,j=1,\ldots,N,                                \\
                                   & \sum_{\substack{i=1,\ldots,N                                                                                        \\ p = 0,1}} \tr_2(E_{i,p}) \propto \Id, \\
                                   & E_{i,p} \ge 0 \quad \mathrm{for} \quad i=1,\ldots,N \quad \mathrm{and} \quad p=0,1,                                 \\[2mm]
            \text{variables:}\quad & E_{i,p} \in \mathcal{M}_{4\times4}(\mathbb{C}) \quad \mathrm{for} \quad i=1,\ldots,N \quad \mathrm{and} \quad p=0,1
        \end{split}
    \end{equation*}
    \captionsetup{name=SDP}
    \caption{Semidefinite program for computing $y_{\mathcal{Q}, P}$. The function $p(a,b)$ equals $p(a,b)=a$ if $P=A$ and $p(a,b)=b$ if $P=B$.}
    \label{sdp-2}
\end{table}

Using SDP \ref{sdp-1} and SDP \ref{sdp-2}, we numerically evaluate the functions $x_{\mathcal Q}(\epsilon)$ and $y_{\mathcal Q}(\epsilon)$ for several QKD protocols.
Figure \ref{fig-1} shows the comparison between the key agreement probability of Alice and Bob $x_{\mathcal Q}(\epsilon)$ and the maximum key agreement probability achievable by Eve $y_{\mathcal Q}(\epsilon)$ obtained from the SDP optimization. The black line represents the security boundary: points below this line correspond to a positive asymptotic secret key rate, while the intersection of each curve with the diagonal determines the maximal tolerable noise level $\Upsilon_{\mathcal Q}$. The dots on each curve indicate increments of the noise parameter $\epsilon$ by $0.01$.
As we can see, the results reveal a clear hierarchy in noise resilience between the well-known QKD protocols. The six-state protocol achieves the highest threshold, $\Upsilon_{6S}=0.25$, followed by BB84 with $\Upsilon_{\mathrm{BB84}}=0.146$. In contrast, E91 and B92 tolerate significantly lower noise levels, with thresholds of $0.0541$ and $0.038$, respectively.

\begin{figure}[!t]
    \centering
    \includegraphics[width=\columnwidth]{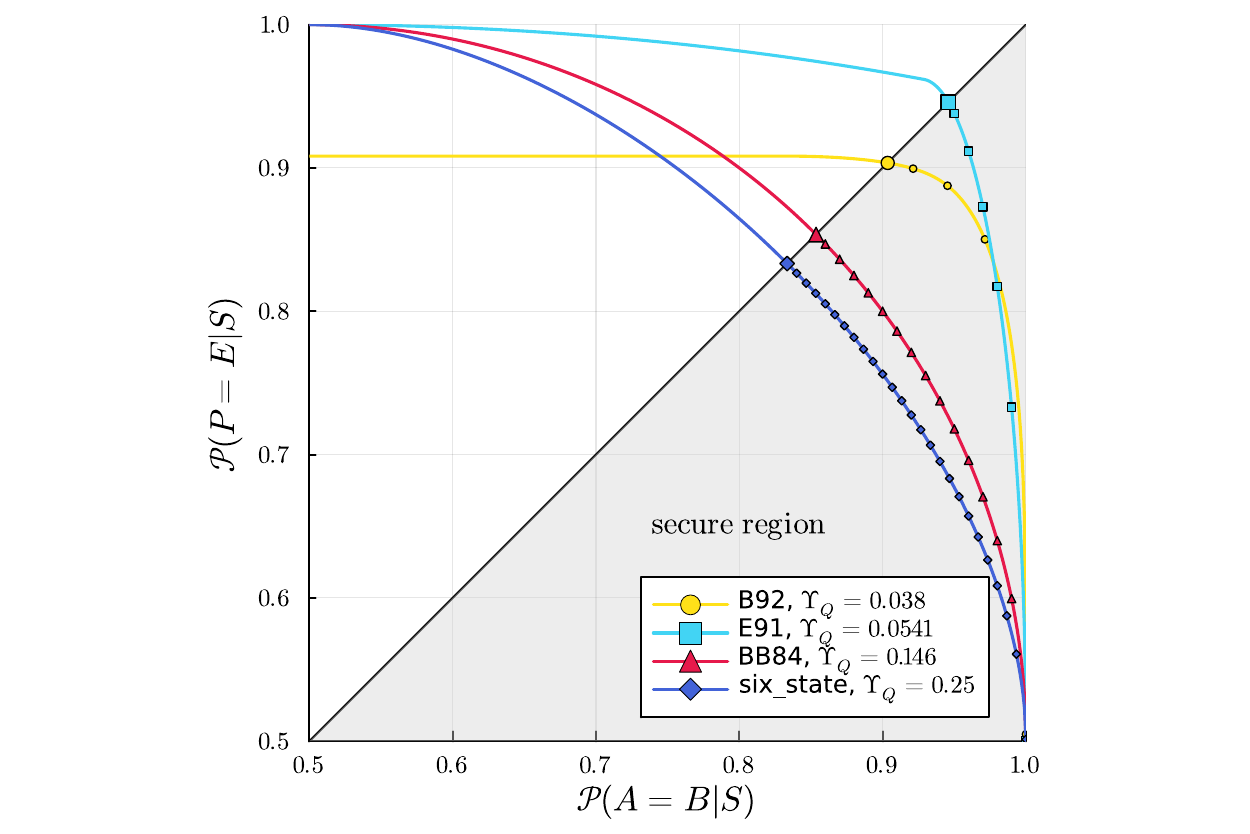}
    \caption{Comparison between the key agreement probability of Alice and Bob and the maximum key agreement probability achievable by Eve for several QKD protocols. The gray field indicates the secured region where $\mathcal{P}(A=B|S) > \mathcal{P}(A=E|S)$. The black diagonal line denotes the security boundary $x_{\mathcal{Q}}(\epsilon) = y_{\mathcal{Q}}(\epsilon)$. Markers correspond to increments of $\epsilon$ by $0.01$ starting with $\epsilon=0$ for $\mathcal{P}(A=B|S)=1$. The larger markers represent the crossing points.}
    \label{fig-1}
\end{figure}

\subsection{QBER Threshold}\label{sub-qber}
In this subsection, we provide an example of the  prepare-and-measure QKD protocol $\QQ^S$ that can tolerate high levels of empirical errors measured in generated bitstrings, yet for which $\Upsilon_{\QQ^S}$ is low.

Let us define a family of QKD protocols $\QQ^S(n,m)$ where $n,m \in \mathbb{N}$ that, for $i=1,\ldots,24$, are defined by
\begin{align*}
    p_{i,0}      & = \frac{m}{24(m+1)}, \qquad p_{i,1} = \frac{1}{24(m+1)}, \\
    \rho_{i,0}   & = U_i\proj{0}U_i^\dagger,                                \\
    \rho_{i,1}   & = U_i[\cos(\pi/(2n)), \sin(\pi/(2n))]^\top               \\
                 & \quad\, [\cos(\pi/(2n)), \sin(\pi/(2n))] U_i^\dagger,    \\
    \Omega_{i,0} & = \frac{1}{M} U_i[-\sin(\pi/(2n)), \cos(\pi/(2n))]^\top  \\
                 & \quad\, [-\sin(\pi/(2n)), \cos(\pi/(2n))] U_i^\dagger,   \\
    \Omega_{i,1} & = \frac{m}{M} U_i\proj{1}U_i^\dagger,
\end{align*}
where $(U_i)_{i=1,\ldots,24}$ is a collection of 24 single-qubit Clifford unitary matrices \cite{Bennett1999Nonlocality} and $M$ is chosen to satisfy $\|\sum_{i,b} \Omega_{i,b} \|_\infty = 1$.

Observe that in this protocol Alice sends the bit $'0'$ $m$ times more frequently than the bit $'1'$. The states are pairwise orthogonal for $n=1$ and are getting closer to each other when $n \to \infty$. Bob performs unambiguous state discrimination with his measurement. Therefore, we have $\PP(A=B|S)=1$ without interference. Moreover, he post-processes received bits and accepts $m$ times more frequently bit $'1'$, than $'0'$. That is an important feature to ensure that $\PP(A=E|S) = \PP(B=E|S) = 1/2$ without interference. Finally, observe that $\QQ^S(1,1) = \QQ_{6S}$. Here, we are only interested in the relation between QBER and $y_\QQ$. Hence, we will modify our equation accordingly
\begin{equation}
    y^*_{\mathcal{Q}}(x) =
    \min_{P\in \{A,B\} }
    \max_E
    \left\{
    \PP(P=E|S):
    \PP(A=B|S)\ge x
    \right\}.
\end{equation}
The intersection point $x=y^*_\QQ(x)$ will define the threshold for acceptable QBER.

Figure~\ref{fig-2} presents examples of the $\mathcal{Q}^S(n,m)$ protocols obtained via SDP optimization. We observe that $\mathcal{Q}^S(n,m)$ for $n,m>1$ can tolerate a higher QBER than the six-state protocol $\QQ_{6S}=\mathcal{Q}^S(1,1)$. However, its maximal tolerable noise parameter decreases. This observation motivates a further investigation of the asymptotic behavior of the tolerable QBER in the regime $m,n \to \infty$, which is seen in Proposition \ref{prop-2-1}.
\begin{figure}[!t]
    \centering
    \includegraphics[width=\columnwidth]{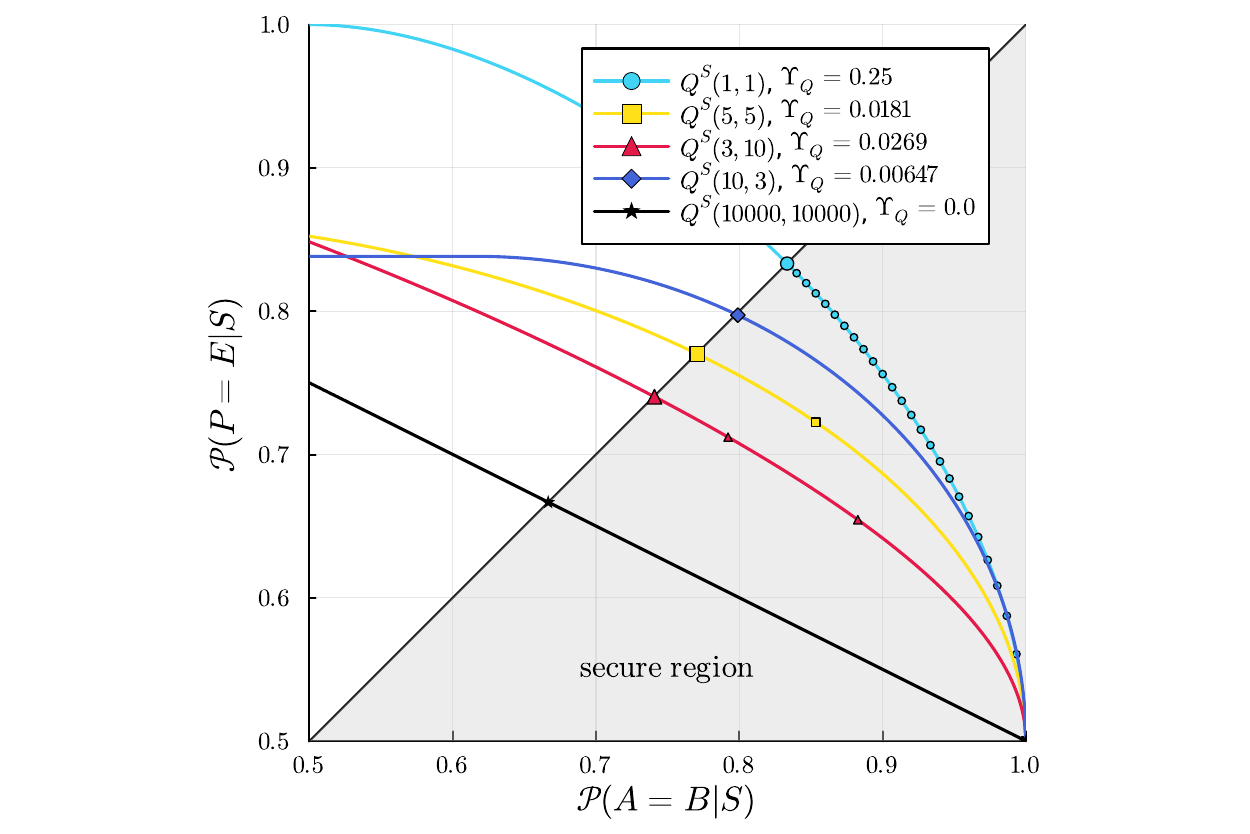}
    \caption{Comparison between the key agreement probability of Alice and Bob and the maximum key agreement probability achievable by Eve for several QKD protocols $\mathcal{Q}^S(n,m)$ to compare with the six-state protocol $\QQ_{6S}=\mathcal{Q}^S(1,1)$.  The gray field indicates the secured region where $\mathcal{P}(A=B|S) > \mathcal{P}(A=E|S)$.  Again, the black diagonal line denotes the security boundary $x_{\mathcal Q}(\epsilon)=y_{\mathcal Q}(\epsilon)$.    Markers correspond to increments of $\epsilon$ by $0.01$ starting with $\epsilon=0$ for $\mathcal{P}(A=B|S)=1$. The larger markers represent the crossing points.}
    \label{fig-2}
\end{figure}

\begin{proposition}\label{prop-2-1}\sloppy
    Let $x_{n,m}$ be the limit of acceptable QBER for $\QQ^S(n,m)$, that is $x_{n,m}=y^*_{\QQ^S(n,m)}(x_{n,m})$. We have
    \begin{equation}
        \lim_{n,m \to \infty} x_{n,m} = \frac{2}{3}.
    \end{equation}
    Moreover, it holds
    \begin{equation}
        \lim_{n,m \to \infty} \Upsilon_{\QQ^S(n,m)} = 0.
    \end{equation}
\end{proposition}
The proof of Proposition~\ref{prop-2-1} is presented in Appendix~\ref{proof-2-1}.

\section{Conclusion and discussion}\label{sec:discuss}

In this work, we analyzed the noise resilience of QKD protocols secured against independent attacks with one-way classical post-processing.
To this end, we introduced an $\epsilon$-based framework that quantifies the implementation noise through the fidelity between the identity channel and the effective quantum channel $\Phi_\epsilon$. The channel $\Phi_\epsilon$ describes the cumulative effect of all implementation imperfections. Within this framework, we defined the fundamental noise threshold $\Upsilon$, representing the maximum implementation noise compatible with secure key generation.

The main analytical result is presented in Theorem~\ref{thm-1-1}. We determined the noise threshold for QKD protocols protected against independent attacks using one-way communication. Furthermore, we showed that the six-state protocol achieves this optimal value. Next, we developed a general semidefinite programming (SDP) framework for evaluating the noise resilience of QKD protocols (SDP \ref{sdp-1} and SDP \ref{sdp-2}). This framework was then applied to several representative protocols, including BB84, B92, and E91, allowing us to compare their robustness using the proposed $\epsilon$-metric.
We also provided a numerical example of a family of prepare-and-measure QKD protocols that can tolerate high levels of empirical errors measured in the generated bitstrings, while exhibiting a low value of $\Upsilon$. This example serves as the main motivation for introducing our framework and demonstrates the distinction between empirical error rates and physical implementation noise. Finally, for this family of protocols, we investigated the asymptotic behavior of the tolerable QBER presented in Proposition \ref{prop-2-1}.

Several directions remain open for future investigation. An immediate extension would be to replace the fidelity-based definition of $\epsilon$ with alternative measures, such as the diamond norm \cite{Watrous2009} or other operationally motivated metrics. Since both of those measures are unitary invariant, we suspect that the optimality of the six-state protocol persists under unitary invariant measures.

Another natural direction is to extend the present analysis beyond independent attacks. In particular, it would be valuable to investigate the proposed framework for protocols employing advantage distillation, as well as under collective and coherent attacks. Finally, extending our framework to measurement-device-independent \cite{Lo2012} and fully device-independent QKD scenarios \cite{Vazirani2014, Zapatero2023} could provide further insight into the fundamental limitations imposed by implementation noise in more general cryptographic settings.

\appendices
\section{Proof of Theorem~\ref{thm-1-1}}\label{proof-1-1}
\begin{proof}
    Let us consider a general QKD protocol $\QQ$, as defined in Section~\ref{sec:problem} and depicted in Figure~\ref{fig-proof-1-1-1}.

    \begin{figure}[htp!]
        \centering
        \includegraphics[width=0.5\linewidth]{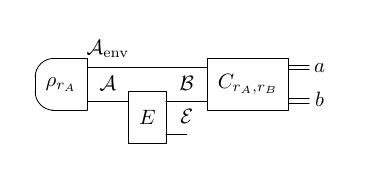}
        \caption{The general setup of a single round of the QKD protocol $\QQ$ defined in Section~\ref{sec:problem}. Alice prepares the state $\rho_{r_A}$ on registers $\mathcal{A} \otimes \mathcal{A}_{\mathrm{env}}$, Eve attacks the transmitted qubit with the channel $E$, and finally
            the parties Alice and Bob implement the LOCC measurement $C_{r_A,r_B}$ producing the key bits $a$ and $b$.}
        \label{fig-proof-1-1-1}
    \end{figure}

    Alice and Bob begin by generating their secret random variables, $r_A$ and $r_B$, respectively. Alice then prepares a quantum state $\rho_{r_A}$ defined on the composite system $\mathcal{A} \otimes \mathcal{A}_{\mathrm{env}}$, where $\mathcal{A}=\mathbb{C}^2$ is a qubit register whose state is sent to Bob, and $\mathcal{A}_{\mathrm{env}}$ is Alice's private quantum register. The state acting on $\mathcal{A}$ is intercepted by Eve, who performs the channel $E$ and sends a qubit state on register $\mathcal{B}$ to Bob. The state on Eve's register $\mathcal{E}$  is kept for later usage to generate her bit. At the final stage by local operations and classical communication Alice and Bob implement the LOCC measurement \cite{Bennett1999Nonlocality} $C_{r_A,r_B}$ based on their random variables. As a result, the round is either successful $S$ or discarded. If it succeeds, then Alice measures key bit $a$ and Bob key bit $b$.

    Observe that the value $\Upsilon_\QQ$ is independent of the value of $\PP(S)$, therefore, w.l.o.g. we can arbitrarily reduce the success probability of any $\QQ$ to simplify the problem. In particular, the LOCC measurement $C_{r_A, r_B}$ is a separable measurement \cite{Bennett1999Nonlocality},  and it can be simulated by performing independent local measurements followed by post-selection \cite{watrous2018theory} resulting only in the drop of $\PP(S)$. Therefore, the distribution of $a$ for Alice depends only on $r_A$ and hence, she does not need the $\mathcal{A}_{\mathrm{env}}$ register. Based on that, $\QQ$ will be a QKD scheme, where Alice produces states $\rho_{r_A}$ on $\AA$ and sends them to Bob. His local measurement (depending on $r_B$) can be combined into a single POVM $\Omega$ with a larger number of labels. Let $o_B$ be the outcome of the measurement $\Omega$. After the quantum communication stage, Alice and Bob perform classical post-processing to reconcile their outcomes. The goal of this step is to implement a deterministic decision function that takes $r_A$ and $o_B$ and maps it to the keys $a,b$, respectively,  if the round is successful $S$. Recall that we can reduce $\PP(S)$ by implementing the communication stage in a probabilistic way, so that for each $r_A$ Alice accepts only a particular, unique set of outputs $o_B$.

    \begin{figure}[hpt!]
        \centering
        \includegraphics[width=0.75\linewidth]{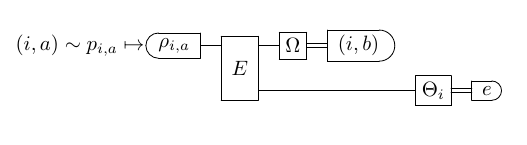}
        \caption{The simplified prepare-and-measure form of the QKD protocol $\QQ$ obtained in the proof of Theorem~\ref{thm-1-1}. Alice sends the state $\rho_{i,a}$ to Bob with probability $p_{i,a}$, Bob performs the measurement $\{\Omega_{i,b}\}_{i,b}$, and Eve applies the channel
            $E$ followed by the conditional measurement $\Theta_{i}$, performed after the index $i$ is publicly announced.}
        \label{fig-proof-1-1-2}
    \end{figure}

    Within this simplification, which can be seen in Figure \ref{fig-proof-1-1-2}, we define $\QQ$ to contain Alice's probabilities $
        (p_{i,a})_{\substack{i=1,\ldots,N\\a=0,1}}, $
    where \(p_{i,a}\geq 0\), together with a corresponding set of qubit states
    $ \{\rho_{i,a}\}_{\substack{i=1,\ldots,N\\a=0,1}}. $
    We assume that for every $i$ we have  $\sum_{a=0}^{1} p_{i,a} > 0 $.
    Bob is equipped with a sub-normalized measurement $
        \{\Omega_{i,b}\}_{\substack{i=1,\ldots,N\\b=0,1}}, $
    where \(\Omega_{i,b}\geq 0\) for all \(i\) and \(b\), and such that
    $
        \sum_{b=0}^{1}\Omega_{i,b} > 0
    $ for every \(i\).  Alice and Bob accept the round whenever their measurement outcomes correspond to the same index $i$.
    Meanwhile, Eve intercepts the quantum communication with channel $E$ and then by hearing label $i$ performs the conditional measurement $\Theta_i$ to obtain her bit $e$. In total, actions of Eve can be described by a quantum network \cite{Bisio2011} via operators $E_{i,e} \ge 0$ for $i=1,\ldots,N$ and $e=0,1$ such that $\sum_e E_{i,e} = \sum_e E_{j,e}$ for any $i,j$ and $\tr_2\left( \sum_e E_{1,e} \right) = \Id$.

    Let us analyze $x_\QQ(\epsilon)$ and $y_\QQ(\epsilon)$ within the provided framework. For $x_\QQ$ we have
    \begin{equation}
        \begin{split}
            \PP(S)            & = \sum_{\substack{i=1,\ldots,N \\a,b=0,1}} p_{i,a} \tr\left( J(\Phi_\epsilon) \overline{\rho_{i,a}} \otimes \Omega_{i,b}\right);\\
            \PP(A=B \wedge S) & = \sum_{\substack{i=1,\ldots,N \\a=0,1}} p_{i,a} \tr\left( J(\Phi_\epsilon) \overline{\rho_{i,a}} \otimes \Omega_{i,a}\right).
        \end{split}
    \end{equation}
    The SDP problem for $x_\QQ$ is obtained through
    \begin{equation}
        \begin{split}
            x_\QQ(\epsilon)
            = & \min \left\{\frac{\PP(A=B \wedge S)}{\PP(S)}: \frac14 \braV{\Id}J(\Phi_\epsilon)\ketV{\Id}\ge(1-\epsilon),  \tr_2(J(\Phi_\epsilon)) = \Id, J(\Phi_\epsilon) \ge 0 \right\}                     \\
            = & \min \left\{\frac{\PP(A=B \wedge S)}{\PP(S)}: \braV{\Id}J(\Phi_\epsilon)\ketV{\Id}\ge\tr(J(\Phi_\epsilon))(2-2\epsilon),  \tr_2(J(\Phi_\epsilon)) \propto \Id, J(\Phi_\epsilon) \ge 0 \right\} \\
            = & \min \Bigg\{\PP(A=B \wedge S):                                                                                                                                                                 \\& \PP(S)=1, \braV{\Id}J(\Phi_\epsilon)\ketV{\Id}\ge\tr(J(\Phi_\epsilon))(2-2\epsilon),  \tr_2(J(\Phi_\epsilon)) \propto \Id, J(\Phi_\epsilon) \ge 0 \Bigg\}.
        \end{split}
    \end{equation}

    For $y_\QQ(\epsilon)=\min\{y_{\QQ,A}(\epsilon),y_{\QQ,B}(\epsilon)\}$ we define
    \begin{equation}
        \begin{split}
            \PP(S)            & = \sum_{\substack{i=1,\ldots,N \\a,b,e=0,1}} p_{i,a} \tr\left( E_{i,e} \overline{\rho_{i,a}} \otimes \Omega_{i,b}\right);\\
            \PP(A=B \wedge S) & = \sum_{\substack{i=1,\ldots,N \\a,e=0,1}} p_{i,a} \tr\left( E_{i,e} \overline{\rho_{i,a}} \otimes \Omega_{i,a}\right);\\
            \PP(P=E \wedge S) & = \sum_{\substack{i=1,\ldots,N \\a,b=0,1}} p_{i,a} \tr\left( E_{i,p(a,b)} \overline{\rho_{i,a}} \otimes \Omega_{i,b}\right),\\
        \end{split}
    \end{equation}
    where $p(a,b) = a$ if $P=A$ and $p(a,b) = b$ if $P=B$.
    Then, similarly to $x_\QQ(\epsilon)$ we can write
    \begin{equation}
        \begin{split}
            y_{\QQ,P}(\epsilon)
            = & \max \Bigg\{\frac{\PP(P=E \wedge S)}{\PP(S)}: \\& \frac{\PP(A=B \wedge S)}{\PP(S)} \ge x_\QQ(\epsilon), E_{i,e} \ge 0, \sum_e E_{i,e} = \sum_e E_{j,e}, \tr_2\left( \sum_e E_{1,e} \right) = \Id \Bigg\} \\
            = & \max \Bigg\{\frac{\PP(P=E \wedge S)}{\PP(S)}: \\& \frac{\PP(A=B \wedge S)}{\PP(S)} \ge x_\QQ(\epsilon), E_{i,e} \ge 0, \sum_e E_{i,e} = \sum_e E_{j,e}, \tr_2\left( \sum_e E_{1,e} \right) \propto \Id \Bigg\} \\
            = & \max \Bigg\{\PP(P=E \wedge S):                \\& \PP(S)=1, \PP(A=B \wedge S) \ge x_\QQ(\epsilon), E_{i,e} \ge 0, \sum_e E_{i,e} = \sum_e E_{j,e}, \tr_2\left( \sum_e E_{1,e} \right) \propto \Id \Bigg\}.
        \end{split}
    \end{equation}

    The next important step is to observe the monotonicity of the problem. If $\epsilon' < \epsilon$, then naturally $x_\QQ(\epsilon') \ge x_\QQ(\epsilon)$. Moreover, it implies $y_\QQ(\epsilon') \le y_\QQ(\epsilon)$. Hence, for each $\QQ$ if $x_\QQ(\epsilon) > y_\QQ(\epsilon)$, then $x_\QQ(\epsilon') > y_\QQ(\epsilon')$. Let us define a new QKD protocol $\QQ'$ based on $\QQ$. Let $(U_j)_{j=1,\ldots,24}$ be a collection of distinct, single-qubit elements of the Clifford group, which form a unitary $3$-design \cite{Zhu2017, Gross2007}. For $\QQ'$ we define the probability value $p_{i,j,a} = \frac{1}{24}p_{i,a}$, quantum state $\rho_{i,j,a} = U_j \rho_{i,a} U_j^\dagger$, and quantum measurement $\Omega_{i,j,b} = \frac{1}{24} U_{j} \Omega_{i,b} U_j^\dagger$. Let us take $\epsilon$ such that $y_{\QQ,P}(\epsilon) < x_\QQ(\epsilon)$. Then, for $\QQ'$ we have
    \begin{equation}
        \begin{split}
            \PP(S)            & = \sum_{\substack{i=1,\ldots,N              \\j=1,\ldots,24\\a,b=0,1}} p_{i,j,a} \tr\left( J(\Phi_\epsilon) \overline{\rho_{i,j,a}} \otimes \Omega_{i,j,b}\right) \\& = \frac{1}{24} \sum_{\substack{i=1,\ldots,N\\a,b=0,1}} p_{i,a} \tr\left( \sum_{j=1,\ldots,24} (U_j^\top \otimes U_j^\dagger)J(\Phi_\epsilon)(\overline{U_j} \otimes U_j) \overline{\rho_{i,a}} \otimes \Omega_{i,b}\right); \\
            \PP(A=B \wedge S) & = \frac{1}{24} \sum_{\substack{i=1,\ldots,N \\a=0,1}} p_{i,a} \tr\left( \sum_{j=1,\ldots,24} (U_j^\top \otimes U_j^\dagger)J(\Phi_\epsilon)(\overline{U_j} \otimes U_j) \overline{\rho_{i,a}} \otimes \Omega_{i,a}\right).
        \end{split}
    \end{equation}
    Let us now define \begin{equation}
        J(\Phi'_\epsilon) = \frac{1}{24}\sum_{j=1,\ldots,24} (U_j^\top \otimes U_j^\dagger)J(\Phi_\epsilon)(\overline{U_j} \otimes U_j),
    \end{equation}  and observe that $\Phi'_\epsilon$ satisfies the conditions for the optimization domain. Therefore, $x_{\QQ'}(\epsilon) \ge x_\QQ(\epsilon)$. Similarly, for $y_{\QQ',P}(\epsilon)$ we get
    \begin{equation}
        \begin{split}
            \PP(S)            & = \frac{1}{24}\sum_{\substack{i=1,\ldots,N \\a,b,e=0,1}} p_{i,a} \tr\left( \sum_{j=1,\ldots,24} (U_j^\top \otimes U_j^\dagger)E_{i,j,e}(\overline{U_j} \otimes U_j) \overline{\rho_{i,a}} \otimes \Omega_{i,b}\right);\\
            \PP(A=B \wedge S) & = \frac{1}{24}\sum_{\substack{i=1,\ldots,N \\a,e=0,1}} p_{i,a} \tr\left( \sum_{j=1,\ldots,24} (U_j^\top \otimes U_j^\dagger)E_{i,j,e}(\overline{U_j} \otimes U_j) \overline{\rho_{i,a}} \otimes \Omega_{i,a}\right);\\
            \PP(P=E \wedge S) & = \frac{1}{24}\sum_{\substack{i=1,\ldots,N \\a,b=0,1}} p_{i,a} \tr\left(  \sum_{j=1,\ldots,24} (U_j^\top \otimes U_j^\dagger)E_{i,j,p(a,b)}(\overline{U_j} \otimes U_j)  \overline{\rho_{i,a}} \otimes \Omega_{i,b}\right).
        \end{split}
    \end{equation}
    If we define $E'_{i,e} = \frac{1}{24}\sum_{j=1,\ldots,24} (U_j^\top \otimes U_j^\dagger)E_{i,j,e}(\overline{U_j} \otimes U_j)$ we see that it belongs to the optimization domain with restricted structure and an even stronger condition $\PP(A=B|S) \ge x_{\QQ'}(\epsilon) \ge x_\QQ(\epsilon)$. Therefore, we have $y_{\QQ',P}(\epsilon)\le y_{\QQ,P}(\epsilon)$.

    We have shown that $\Upsilon_{\QQ'} \ge \Upsilon_{\QQ}$. For now, we refer to $\QQ$  as the QKD protocol averaged over the Clifford group. Next, we translate $x_\QQ(0) = 1$ and $y_{\QQ, P}(0) = 1/2$ into the definition of $\QQ$. For the noiseless scenario $\PP(A=B|S)=1$,  we have to have $\tr(p_{i,a}\rho_{i,a}\Omega_{i,a \oplus 1}) = 0$ for all $i,a$. For $y_{\QQ, P}(0)$, the inequality $\PP(A=B|S)\ge x_\QQ(0)=1$ requires $\sum_e E_{i,e} = \projV{\Id}$. To ensure  $\PP(P=E|S)=1/2$, the a priori distribution of the key bit must be balanced. That means we have $\tr(p_{i,0}\rho_{i,0}\Omega_{i,0})=\tr(p_{i,1}\rho_{i,1}\Omega_{i,1})$ for all $i$. To evaluate $\Upsilon$,  we observe that
    \begin{equation}
        \Upsilon_\QQ = \max\left(\sup\left\{\epsilon\in[0,1]:y_{\mathcal{Q},A}(\epsilon)< x_{\mathcal{Q}}(\epsilon)\right\}, \sup\left\{\epsilon\in[0,1]:y_{\mathcal{Q},B}(\epsilon)< x_{\mathcal{Q}}(\epsilon)\right\}\right).\end{equation}
    To compute $x_\QQ(\epsilon)$, we use the unitary $t$-design property
    \begin{equation}
        \frac{1}{24}\sum_{j=1,\ldots,24} (U_j^\top \otimes U_j^\dagger)J(\overline{U_j} \otimes U_j) = \int_{U(2)} (U \otimes \overline{U}) J (U \otimes \overline{U})^\dagger dU = (1-q) \projV{\Id_2} + q \Id_4/2,
    \end{equation}
    where $q \in [0, 4/3]$. The averaged noise channel must satisfy \begin{equation}
        \frac14 \braV{\Id}((1-q) \projV{\Id_2} + q \Id_4/2)\ketV{\Id} \ge 1-\epsilon,
    \end{equation} which translates to $\frac43 \epsilon \ge q$. Hence, $q \in [0,4/3 \epsilon]$. The objective value is monotonically decreasing with $q$, hence the solution is reached for $q_\epsilon =4/3 \epsilon$. Similarly for $y_{\QQ, P}(\epsilon)$, we show that Eve's attacks must satisfy \begin{equation}
        \sum_{e} E_{i,e} = (1-q) \projV{\Id_2} + q \Id_4/2
    \end{equation}  due to unitary symmetry. The inequality $\PP(A=B|S) \ge x_\QQ(\epsilon)$ leads to the optimization domain $q \in [0,4/3\epsilon]$. Let $J_{q_\epsilon'}$ be the optimal Choi operator for $q_\epsilon' \in [0,q_\epsilon]$, then $E_{i,e} = \sqrt{J_{q_\epsilon'}} M_{i,e} \sqrt{J_{q_\epsilon'}}$, where $(M_{i,e})_e$ is a conditional POVM with $\sum_e M_{i,e} = \Id$.
    The function $\PP(P=E \wedge S)$ is maximized for the measurements satisfying the Holevo–Helstrom theorem \cite{helstrom1969quantum, watrous2018theory}, explicitly
    \begin{equation}
        \begin{split}
              & \PP(A=E \wedge S)                     \\ =& \sum_{\substack{i=1,\ldots,N\\a=0,1}} \tr\left( M_{i,a} p_{i,a} \sqrt{J_{q_\epsilon}} \left(\overline{\rho_{i,a}} \otimes \sum_{b=0,1} \Omega_{i,b} \right)\sqrt{J_{q_\epsilon}}\right)\\
            = & \frac12  \sum_{\substack{i=1,\ldots,N \\a,b=0,1}} p_{i,a} \tr\left( J_{q_\epsilon} \overline{\rho_{i,a}} \otimes  \Omega_{i,b} \right)
            + \frac12 \sum_{i=1,\ldots,N}
            \left\| \sqrt{J_{q_\epsilon}} \left( (p_{i,0}\overline{\rho_{i,0}} -  p_{i,1}\overline{\rho_{i,1}})\otimes \sum_{b=0,1} \Omega_{i,b} \right)\sqrt{J_{q_\epsilon}}\right\|_1.
        \end{split}
    \end{equation}
    The analogous equation holds for $\PP(B=E \wedge S)$.

    We have evaluated $x_\QQ(\epsilon) = \PP_{q_\epsilon}(A=B|S)$ and $y_{\QQ,P}(\epsilon) = \PP_{q_\epsilon'}(P=E|S)$, where probabilities $\PP_{q}$ are defined with the notation of the channel $J_{q}$ for both events, $A=B$ and $P=E$. Here, the parameter $N$ defines the number of pairs of strategies. In total, due to symmetrization there are $24N$ pairs for $\QQ$. As $J_{q}$ commutes with $U_j \otimes \overline{U_j}$ the results are multiplied on both sides by $24\PP_{q}$. Hence, in that convention we can manipulate $N$ freely.

    To show that $\Upsilon = 1/4$, first, let us define the following QKD with $N=1$ given by $p_{1,a} = 1/2$, $\rho_{1,a} = \proj{a}$ and $\Omega_{1,b}= \proj{b}$. Observe that $\PP(S) = 1$ for all values of $q$. In particular $x_\QQ(\epsilon) = \PP_{q_\epsilon}(A=B|S) = \PP_{q_\epsilon}(A=B \wedge S) = 1 - \frac{q_\epsilon}{2} = 1 - \frac{2}{3}\epsilon$. Due to symmetry of the problem we compute only $\PP_{q_\epsilon'}(A=E|S)$. Here,
    \begin{equation}
        \begin{split}
            \PP_q(A=E \wedge S) & = \frac12  \sum_{a,b=0,1} \frac12 \tr\left( J_{q} \proj{a} \otimes  \proj{b} \right)
            + \frac14 \left\| \sqrt{J_{q}} \left( Z \otimes \Id \right)\sqrt{J_{q}}\right\|_1                          \\
                                & =\frac12 + \frac14 (q+\sqrt{4q-3q^2}),
        \end{split}
    \end{equation}
    where $q \le \frac{4}{3}\epsilon$. If $\epsilon \le \frac{1}{2}$, then $\PP_q(A=E \wedge S)$ is increasing with $q$, hence it is saturated for $q = q_\epsilon$ and equals $y_{\mathcal{Q},A}(\epsilon) = \PP_{q_\epsilon}(A=E \wedge S) = \frac12 + \frac13 \epsilon + \frac{1}{\sqrt{3}}\sqrt{\epsilon}\sqrt{1-\epsilon}$. Then, $y_{\mathcal{Q},A}(\epsilon)< x_{\mathcal{Q}}(\epsilon)$ if and only if $\epsilon < \frac14$. Therefore, for this protocol we obtain
    \begin{equation}
        \Upsilon_\QQ = \frac{1}{4},
    \end{equation}
    which establishes the lower bound $\Upsilon \ge \frac14$. It remains to prove the matching upper bound, that is, $\Upsilon_\QQ \le \frac{1}{4}$ for any $\QQ$. For any $\QQ$ we lower-bound $y_{\QQ,P}(\epsilon) \ge \PP_{q_\epsilon}(P=E|S)$. Therefore, we have
    \begin{equation}
        \begin{split}
            \Upsilon_\QQ & = \max_{P=A,B}\left(\sup\left\{\epsilon\in[0,1]:y_{\mathcal{Q},P}(\epsilon)< x_{\mathcal{Q}}(\epsilon)\right\}\right) \\ &\le \max_{P=A,B}\left(\sup\left\{\epsilon\in[0,1]: \PP_{q_\epsilon}(P=E \wedge S) < \PP_{q_\epsilon}(A=B \wedge S)\right\}\right)\\ &= \max_{P=A,B}\left(\sup\left\{\epsilon\in[0,1]: \sum_{i=1,\ldots,N} \PP_{q_\epsilon}(P=E \wedge S \wedge i) < \sum_{i=1,\ldots,N} \PP_{q_\epsilon}(A=B \wedge S \wedge i)\right\}\right).
        \end{split}
    \end{equation}
    For $\QQ$ let $P$ be the person for whom the maximum of $\Upsilon_\QQ$ is reached. Let $(\epsilon_n)_{n \in \mathbb{N}}$ be a sequence for which $\epsilon_n \to \Upsilon_\QQ$ and \begin{equation}
        \sum_{i=1,\ldots,N} \PP_{q_{\epsilon_n}}(P=E \wedge S \wedge i) < \sum_{i=1,\ldots,N} \PP_{q_{\epsilon_n}}(A=B \wedge S \wedge i).
    \end{equation} From the Pigeonhole principle,  there is $i_0$ for which $\PP_{q_{\epsilon_{n_k}}}(P=E \wedge S \wedge i_0) < \PP_{q_{\epsilon_{n_k}}}(A=B \wedge S \wedge i_0)$, where $\epsilon_{n_k}$ is an infinite subsequence. Therefore, to optimize $\Upsilon_\QQ$ we take $N=1$. Let $\sigma_a = p_a \rho_a$. We have
    \begin{equation}
        \begin{split}
            \PP_{q_\epsilon}(A=B \wedge S) & = \sum_{a=0,1} \tr\left( J_{q_\epsilon} \overline{\sigma_{a}} \otimes \Omega_{a}\right)                                              \\
            \PP_{q_\epsilon}(A=E \wedge S) & = \frac12  \sum_{a,b=0,1} \tr\left( J_{q_\epsilon} \overline{\sigma_{a}} \otimes  \Omega_{b} \right)
            + \frac12 \left\| \sqrt{J_{q_\epsilon}} \left( (\overline{\sigma_{0}} -  \overline{\sigma_{1}})\otimes \sum_{b=0,1} \Omega_{b} \right)\sqrt{J_{q_\epsilon}}\right\|_1 \\
            \PP_{q_\epsilon}(B=E \wedge S) & = \frac12  \sum_{a,b=0,1} \tr\left( J_{q_\epsilon} \overline{\sigma_{a}} \otimes  \Omega_{b} \right)
            + \frac12 \left\| \sqrt{J_{q_\epsilon}} \left( \sum_{a=0,1} \overline{\sigma_{a}} \otimes (\Omega_{0} - \Omega_1) \right)\sqrt{J_{q_\epsilon}}\right\|_1.
        \end{split}
    \end{equation}
    Due to the symmetry of the simplified problem with respect to $P$ we get
    \begin{equation}\label{eq-proof-1-1-proj}
        \begin{split}
             & \Upsilon_\QQ \le \sup\left\{\epsilon\in[0,1]: \right.                                                                                                                                                                                                                                                                \\
             & \left. \frac12  \sum_{a,b=0,1} \tr\left( J_{q_\epsilon} \overline{\sigma_{a}} \otimes  \Omega_{b} \right)
            + \frac12 \left\| \sqrt{J_{q_\epsilon}} \left( (\overline{\sigma_{0}} -  \overline{\sigma_{1}})\otimes \sum_{b=0,1} \Omega_{b} \right)\sqrt{J_{q_\epsilon}}\right\|_1 < \sum_{a=0,1} \tr\left( J_{q_\epsilon} \overline{\sigma_{a}} \otimes \Omega_{a}\right)\right\}                                                   \\
             & = \sup\left\{\epsilon\in[0,1]: \left\| \sqrt{J_{q_\epsilon}} \left( (\overline{\sigma_{0}} -  \overline{\sigma_{1}})\otimes \sum_{b=0,1} \Omega_{b} \right)\sqrt{J_{q_\epsilon}}\right\|_1 < \tr\left( J_{q_\epsilon} (\overline{\sigma_{0}} - \overline{\sigma_{1}}) \otimes (\Omega_{0} - \Omega_1)\right)\right\}
        \end{split}
    \end{equation}
    Recall the assumptions $\tr(\sigma_a \Omega_{a \oplus 1}) = 0$ and $\tr(\sigma_0 \Omega_0) = \tr(\sigma_1 \Omega_1)$. Based on $U \otimes \overline{U}$ symmetry in Eq.~\eqref{eq-proof-1-1} we can take the following parametrization of $\QQ$: $\sigma_0 = c\proj{0}$, $\sigma_1 = \proj{\alpha}$, $\Omega_0 = \proj{\alpha^\perp}$, $\Omega_1 = c\proj{1}$, where $c > 0$ and $\alpha \in [0,\pi/2]$, and $\ket{\alpha} = [\cos(\alpha), \sin(\alpha)]$, $\ket{\alpha^\perp} = [-\sin(\alpha), \cos(\alpha)]$.

    Let $\Psi(X) = \Pi_{IZ} X \Pi_{IZ} + \Pi_{XY} X \Pi_{XY}$ be a quantum channel, where $\Pi_{IZ}$ is the projector onto the subspace spanned by $\ketV{\Id}, \ketV{Z}$ and $\Pi_{XY}$ is the projector onto the subspace spanned by $\ketV{X}, \ketV{Y}$. By the data processing inequality we get for $M = \left\| \sqrt{J_{q_\epsilon}} \left( (\overline{\sigma_{0}} -  \overline{\sigma_{1}})\otimes \sum_{b=0,1} \Omega_{b} \right)\sqrt{J_{q_\epsilon}}\right\|_1$ that $\|M\|_1 \ge \|\Pi_{IZ}M\Pi_{IZ}\|_1+\|\Pi_{XY}M\Pi_{XY}\|_1$. It translates to the inequality
    \begin{equation}\label{eq-proof-1-1}
        \begin{split}
             & \Upsilon_\QQ \le \sup\left\{\epsilon\in[0,1]: \|\Pi_{IZ}M\Pi_{IZ}\|_1+\|\Pi_{XY}M\Pi_{XY}\|_1 < \tr\left( J_{q_\epsilon} (\overline{\sigma_{0}} - \overline{\sigma_{1}}) \otimes (\Omega_{0} - \Omega_1)\right)\right\}.
        \end{split}
    \end{equation}
    Let us fix $\epsilon = 1/4$. We will show that for any $\alpha,c$ we have $\|\Pi_{IZ}M\Pi_{IZ}\|_1+\|\Pi_{XY}M\Pi_{XY}\|_1 \ge \tr\left( J_{q_\epsilon} (\overline{\sigma_{0}} - \overline{\sigma_{1}}) \otimes (\Omega_{0} - \Omega_1)\right)$ or equivalently that
    \begin{equation}
        \begin{split}
             & \frac{1}{3}\sqrt{\sin^4(\alpha)\cos^4(\alpha)+9c^2\sin^4(\alpha)}+\frac{1}{6}\sqrt{(2\sin^2(\alpha)\cos^2(\alpha)+c^2-1)^2+4c^2\sin^4(\alpha)} \\
             & \ge \frac{4}{3}c \sin^2(\alpha) - \frac{1}{6}(1-c)^2.
        \end{split}
    \end{equation}
    Indeed, direct estimation reveals
    \begin{equation}
        \begin{split}
             & \frac{1}{3}\sqrt{\sin^4(\alpha)\cos^4(\alpha)+9c^2\sin^4(\alpha)}+\frac{1}{6}\sqrt{(2\sin^2(\alpha)\cos^2(\alpha)+c^2-1)^2+4c^2\sin^4(\alpha)} \\ \ge& \frac{1}{3}\sqrt{9c^2\sin^4(\alpha)}+\frac{1}{6}\sqrt{4c^2\sin^4(\alpha)} =  \frac{4}{3}c \sin^2(\alpha) \\ \ge& \frac{4}{3}c \sin^2(\alpha) - \frac{1}{6}(1-c)^2,
        \end{split}
    \end{equation}
    which proves the inequality $\Upsilon \le 1/4$.
\end{proof}

\section{Proof of Corollary~\ref{cor-1-2}}\label{proof-1-2}
\begin{proof}
    Investigating Proof~\ref{proof-1-1} backwards we get that to optimize $\Upsilon_\QQ$ over $\QQ$ we take $p_a = 1/2$, $\rho_a = \proj{a}$ and $\Omega_b = \proj{b}$. We have $N=1$, but we average the input states $\rho$ and effects $\Omega$ over Clifford unitary operations $(U_j)_{j=1,\ldots,24}$. Direct computations reveal we end up with three bases $\rho_{1,a} = \proj{a}, \rho_{2,a} = \proj{+_a}$ and $ \rho_{3,a} = \proj{i_a}$ (the same for the measurement $\Omega$). Each of the bases appears $8$ times (up to bit permutation). Therefore, the optimal protocol is equivalent to $\QQ_{6S}$.
\end{proof}

\section{Proof of Proposition~\ref{prop-2-1}}\label{proof-2-1}
\begin{proof}
    We apply Proof~\ref{proof-1-1}. The optimal Eve attack is given by $J_q = (1-q)\projV{\Id} + q/2 \Id$. The parameter $q = q_x$ is set to satisfy $\PP_{q_x}(A=B|S) = x$. In this case, we have  $y^*_\QQ(x) = \PP_{q_x}(A=E|S)$. Let us upper-bound this value
    \begin{equation}
        \begin{split}
             & \PP_{q_x}(A=E \wedge S)                                                                                                                                            \\ & =\frac12  \sum_{a,b=0,1} \tr\left( J_{q_x} \overline{\sigma_{a}} \otimes  \Omega_{b} \right)
            + \frac12 \left\| (a\projV{\Id} + b\Id) \left( (\overline{\sigma_{0}} -  \overline{\sigma_{1}})\otimes \sum_{b=0,1} \Omega_{b} \right)(a\projV{\Id} + b\Id)\right\|_1 \\
             & \le \frac12  \sum_{a,b=0,1} \tr\left( J_{q_x} \overline{\sigma_{a}} \otimes  \Omega_{b} \right)
            + \frac12 \left( 2ab \left\| \ketbraV{M}{\Id}\right\|_1 + b^2 \|\sigma_0 - \sigma_1 \|_1 \|\Omega_0 + \Omega_1\|_1\right),
        \end{split}
    \end{equation}
    where $\sqrt{J_{q_x}} = a\projV{\Id} + b\Id$ and $M = (\sigma_{0} -  \sigma_{1})(\Omega_0 + \Omega_{1})$. For a fixed $n,m$ there is $x_0$, such that it is equal to the upper-bounded $\PP_{q_{x_0}}(A=E | S)$. Then, it holds that
    \begin{equation}
        \begin{split}
            2ab \left\| \ketbraV{M}{\Id}\right\|_1 + b^2 \|\sigma_0 - \sigma_1 \|_1 \|\Omega_0 + \Omega_1\|_1 = \tr\left( J_{q_{x_0}} (\sigma_{0} - \sigma_{1}) \otimes (\Omega_{0} - \Omega_1)\right).
        \end{split}
    \end{equation}
    In the regime $n,m \to \infty$, we solve the equation above to express the behavior of $q_{x_0}$, which is $q_{x_0} = \frac{\pi^2}{2mn^2}(1+o(1))$. That implies
    \begin{equation}
        x_{n,m} \le x_0 = \PP_{q_{x_0}}(A=B|S) = \frac{2}{3}(1+o(1)).
    \end{equation}
    To show the lower bound for $x_{n,m}$ we use the inequality $\|M\|_1 \ge \|\Pi_{IZ}M\Pi_{IZ}\|_1+\|\Pi_{XY}M\Pi_{XY}\|_1$. For it, we repeat the same reasoning which leads to $q_{x_0} = \frac{\pi^2}{2mn^2}(1+o(1))$ and eventually $x_{n,m} \ge x_0 = \PP_{q_{x_0}}(A=B|S) = \frac{2}{3}(1+o(1))$.

    To show that $ \lim_{n,m \to \infty} \Upsilon_{\QQ^S(n,m)} = 0$ we fix an arbitrary $\epsilon > 0$. Then, we fix $q=4/3 \epsilon$ and compute $\PP_{q}(A=B|S) = \frac{2(1-q)m \sin^2(\pi/(2n))+qm}{2(1-q)m \sin^2(\pi/(2n)) + q/2(1+m)^2} \to 0$ for $n,m \to \infty$. Therefore, the most destructive noise channel will have parameter $q' \in (0, q)$ for which $x_{\QQ^S(n,m)}(\epsilon) = \PP_{q'}(A=B|S) = 1/2$ for $n \ge n_0$ and $m \ge m_0$. We have shown that $\Upsilon_{\QQ^S(n,m)} \le \epsilon$ for $n \ge n_0, m\ge m_0$, which finishes the proof.
\end{proof}

\section*{Acknowledgment}
The authors would like to thank Norbert Luchowski for fruitful discussions.

\section*{Funding}
A. B\'ilek is supported by the Ministry of Education, Youth and Sports of the Czech Republic through the e-INFRA CZ (ID:90254). P. Lewandowska is supported by the Ministry of Education, Youth and Sports of the Czech Republic through the e-INFRA CZ (ID:90254), with the financial support of the European Union under the REFRESH -- Research Excellence For REgion Sustainability and High-tech Industries project number CZ.10.03.01/00/22\_003/0000048 via the Operational Programme Just Transition.  R. Kukulski is supported by the European Union under the Quantum error correction codes enhanced by reinforcement learning dedicated for the Ising model-based optimization, contract nr. 01906/2025/RRC via the Operational Programme Just Transition and Moravian-Silesian Region.
    {\L}. Pawela and Z. Pucha{\l}a are supported by the National Science Center (NCN), Poland, under Project Opus No. 2022/47/B/ST6/02380.

\section*{Author Contributions}
All authors contributed equally to this work, including the preparation and writing of the manuscript. R. Kukulski and P. Lewandowska conducted the theoretical security analysis, A. B\'ilek and {\L}. Pawela performed the optimization tasks and generated the numerical results, Z. Pucha{\l}a and {\L}. Pawela defined the research objectives and coordinated project administration.

\section*{Data Availability}
The code that supports the findings of this study is openly available on GitHub~\cite{repository}.

\bibliographystyle{unsrt}
\bibliography{qkd_tit}

\end{document}